\documentclass[10pt]{iopart}
\usepackage{comment,xcolor,fullpage}
\expandafter\let\csname equation*\endcsname\relax

\expandafter\let\csname endequation*\endcsname\relax
\usepackage{amsmath}
\usepackage{mathtools}
\usepackage{amssymb}
\usepackage{amsfonts}
\usepackage{amsthm}
\usepackage{mathrsfs}
\usepackage[all]{xy}
\usepackage[colorlinks=true, citecolor=blue, linkcolor=blue,urlcolor=teal,linktocpage]{hyperref}
\usepackage{tensor}
\usepackage{footnote}
\usepackage{graphicx} 
\usepackage[justification=raggedright,singlelinecheck=false]{caption}
\usepackage{subcaption}

\usepackage{booktabs}

\newtheorem{thm}[subsection]{Theorem}
\newtheorem{lem}[subsection]{Lemma}

\newtheorem{cor}[subsection]{Corollary}
\newtheorem{defn}[subsection]{Definition}
\newtheorem{prop}[subsection]{Proposition}
\newtheorem{rem}[subsection]{Remark}

\newcommand{\dd}{\mathrm{d}}

\newcommand{\beq}{\begin{equation}}
\newcommand{\eeq}{\end{equation}}

\numberwithin{equation}{section}

\def\al{\alpha}
\def\be{\beta}
\def\ga{\gamma}

\def\R{\mathcal{R}}
\def\Q{\mathcal{Q}}
\def\M{\tilde{M}}
\def\h{\tilde{h}}
\def\l{\tilde{\ell}}
\begin{document}

\title{The small sphere limit of quasilocal energy in higher dimensions along lightcone cuts}

\author{Jinzhao Wang}

\address{Institute for Theoretical Physics, ETH 8093 Z\"urich, Switzerland}
\ead{jinzwang@phys.ethz.ch}

\begin{abstract}
The problem of quasilocal energy has been extensively studied mainly in four dimensions. Here we report results regarding the quasilocal energy in spacetime dimension $n\geq 4$. After generalising three distinct quasilocal energy definitions to higher dimensions under appropriate assumptions, we evaluate their small sphere limits along lightcone cuts shrinking towards the lightcone vertex. The results in vacuum are conveniently represented in terms of the electromagnetic decompositions of the Weyl tensor. We find that the limits at presence of matter yield the stress tensor as expected, but the vacuum limits are in general not proportional to the Bel-Robinson superenergy $Q$ in dimensions $n>4$. The result defies the role of the Bel-Robinson superenergy as characterising the gravitational energy in higher dimensions, albeit the fact that it uniquely generalises. Surprisingly, the Hawking energy and the Brown-York energy exactly agree upon the small sphere limits across all dimensions. The ``new'' vacuum limit $\Q$, however, cannot be interpreted as a gravitational energy because of its non-positivity. Furthermore, we also give the small sphere limits of the Kijowski-Epp-Liu-Yau type energy in higher dimensions, and again we see $\Q$ in place of $Q$. Our work extends earlier investigations of the small sphere limits~\cite{horowitz1982note,bergqvist1994energy,brown1999canonical,yu2007limiting}, and also complements~\cite{miao2017quasi}. 
\end{abstract}

\section{Introduction.}
The gravitational field itself carries energy, but it is tricky to locally describe it in general relativity. It is well- known that the equivalence principle forbids a covariant stress tensor characterising the energy content of the gravitational field \cite{misner2017gravitation}. Nevertheless, there is no obstruction in giving nonlocal prescriptions and the quasilocal energy (QLE) is such an attempt. Over the past half-century, QLE is an ongoing research subject studied by both physicists and mathematicians \cite{szabados2009quasi,hawking1968gravitational,hayward1994quasilocal,penrose1982quasi,brown1993quasilocal,kijowski1997simple,liu2003positivity,wang2009quasilocal,epp2000angular,bartnik1989new,bousso2019outer}. Nevertheless, QLE is rarely studied in spacetime dimensions beyond four. 

Here we make an attempt to investigate the local behaviours, also known as the small sphere limit, of quasilocal energy proposals that can be reasonably generalised to higher dimensions. As the gravitational energy density is an invalid notion, the small sphere limit is as local as one can probe about the gravitational energy. It also serves as an important guidance for a sound definition for QLE. Physically, the limit should be proportional to the stress tensor at leading order or the Bel-Robinson (BR) superenergy $Q_0$ in vacuum \cite{szabados2009quasi}. Given that the BR superenergy uniquely generalises to higher dimensions, which we denote as $Q$, one should expect the QLE defined for higher dimensions to reproduce these small sphere limits universally. In $n=4$, there are many results concerning the small sphere limits of various QLE's. Some notable results are for the Hawking energy (by Horowitz and Schmidt \cite{horowitz1982note}), the Brown-York (BY) energy (by Brown, Lau and York (BLY) \cite{brown1999canonical}), the Kijowski-Epp-Liu-Yau (KELY) energy (by Yu \cite{yu2007limiting}) and the Wang-Yau (WY) energy (by Chen, Wang and Yau \cite{chen2018evaluating}). They all exactly agree upon the non-vacuum limit. In vacuum, these QLE's give $\frac1{90}Q_0$ (up to an extra term in cases of KELY and WY) in the small sphere limit via the lightcone cuts.  We generalise these definitions, except for WY energy, to higher dimensions under appropriate assumptions and study their small-sphere behaviours. We find that when $n>4$, a new quantity $\Q$ rather than the BR superenergy $Q$ plays the role $Q_0$ in four dimensions. In terms of the electromagnetic decomposition of the Weyl tensor, $\Q$ is defined as
\beq\label{newquantity}
\Q :=  \frac{(6n^2-20n+8)E^2+6(n-3)H^2-3D^2}{36(n-3)(n-2)(n^2-1)},
\eeq
which only matches with $Q$ when $n=4.$

There is a canonical way to evaluate the small sphere limits as proposed by Horowitz and Schmidt in studying the Hawking energy~\cite{horowitz1982note}. The small sphere limit towards some point $p$ is taken in the following way. Let $N_p$ denote the future-directed lightcone generated by null generators $\ell^+$ parameterised by affine parameter $l$. We pick a future-directed timelike unit vector $e_0$ and normalised $\ell^+$ at $p$ by 
\beq
\langle e_0,\ell^+ \rangle =-1.\label{normalisation1}
\eeq
The lightcone cut is the family of codimension-two surfaces $S_l$ define as the level sets of $l$ on $N_p$. The ingoing null generators on $N_p$ are denoted as $\ell^-$ and they are normalised by
\beq
\langle \ell^-,\ell^+ \rangle =-1.\label{normalisation2}
\eeq
The small sphere limits are given by evaluating the QLE on $S_l$ and take $l$ to zero.  Note that people use a different small sphere limit in the Riemannian setting \cite{fan2009large,wiygul2018bartnik} and obtain results of the Hawking energy and the BY energy, which are not comparable with results evaluated using the lightcone cuts in the spacetime setting.

Here we first study a natural $n-$dimensional generalisation of the Hawking-Hayward (HH) energy \cite{hawking1968gravitational,hayward1994quasilocal}. Hawking's proposal is motivated by the gravitational radiation and Hayward later refined it to ensure the rigidity holds. We choose to study the HH energy not only because it is a canonical QLE that serves as a useful tool in mathematical relativity, most notably in establishing the Riemannian Penrose inequality \cite{huisken2001inverse}, but also because it does admit a straightforward and rather unique generalisation to higher dimensions. 

We also study a different class of QLE definitions based upon the Hamilton-Jacobi analysis\cite{brown1993quasilocal,liu2003positivity,wang2009quasilocal,miao2015quasi}. They are the Brown-York energy \cite{brown1993quasilocal} and the Kijowski-Epp-Liu-Yau energy \cite{epp2000angular,kijowski1997simple,liu2003positivity}. One important feature that distinguishes this approach from others is that it requires a flat reference via isometric embedding of the codimension-2 surface to the Minkowski spacetime as the zero-point energy. In four dimensions, the existence of the isometric embedding is guaranteed by the results by Nirenberg \cite{nirenberg1953weyl} and Pogorelov \cite{pogorelov1952regularity}. In particular, we will be considering the lightcone reference, where $S$ is embedded on a lightcone in the Minkowski spacetime. We choose to use such a reference in order for our results to be comparable to earlier works by BLY and Yu. The existence for the lightcone embedding is guaranteed by Brinkmann's result \cite{brinkmann1923riemann}. It states that a simply connected $n-2$-dimensional Riemannian manifold $\Sigma$ can be isometrically embedded into a lightcone in the $n$- dimensional Minkowski spacetime if and only if $\Sigma$ is conformally flat .  Furthermore, such embedding is unique up to orthochronous Lorentz transfomrations \cite{liu2018rigidity}. For $n=4$, every two-dimensional small sphere is conformally flat so the existence is guaranteed for any surface.

 In higher dimensions, however, the isometric embedding problem is overdetermined, so generally the reference energy cannot be defined. It is still an open problem to properly generalise the above mentioned Hamilton-Jacobi based proposals to higher dimensions. Nevertheless, for our purposes of looking at the small sphere limit along the lightcone cut specified by $(p,e_0)$, we can still proceed under the assumption that such isometric embeddings do exist for our choices of $(p,e_0)$. Such existence assumption is also held by Miao, Tam and Xie in \cite{miao2017quasi}.  After all, we are interested in the local behaviours of QLE and their connections to the BR superenergy in higher dimensions, so we simply consider cases when the lightcone reference does exist. By Brinkmann's result, this implies that the lightcone cuts should be conformally flat.  
 
The non-vacuum limit in general is given by the Ricci-related quantities and we shall use the Einstein equation to introduce the stress tensor $T_{ab}$.  The vacuum limit is most conveniently represented in the electromagnetic decomposition ($E,H,D$) of the Weyl tensor. One can refer to Section \ref{sec:em} for details. Our main results are stated in the following theorems,
\begin{thm}\label{thm:hawking}
Let $S_l$ be the family of surfaces shrinking towards $p$ along lightcone cuts defined with respect to $(p,e_0)$ in an $n$-dimensional spacetime, the limits of the Hawking energy as $l$ goes to $0$ are
\begin{enumerate}
\item In non-vacuum,
\beq
\lim_{l\rightarrow 0}l^{-(n-1)}M_{H} (S_l) = \frac{\Omega_{n-2}}{n-1}T(e_0,e_0).\label{hawkingnonvac0}
\eeq
\item In vacuum or the stress tensor $T$ vanishes in an open set containing $p$,
\beq
\lim_{l\rightarrow 0}l^{-(n+1)}M_{H} (S_l) = \Q .\label{hawkingvac0}
\eeq
\end{enumerate}
where the tensors $T,E,H,D$ are evaluated at $p$.
\end{thm}

As a corollary, one obtains the $n=4$ result by Horowitz and Schmidt. We see that the non-vacuum case (\ref{hawkingnonvac0}) agrees with our expectation. The coefficient in front of $T_{00}$ is exactly the volume of a flat unit sphere enclosed by $S_l$, and thus together it characterises the dominant matter energy content within the small sphere. However, the vacuum limit (\ref{hawkingvac0}) is not proportional to the BR superenergy $Q$ in any dimensions $n>4$. Surprisingly, the same limit is obtained for the BY energy in arbitrary dimensions.
\begin{thm}\label{thm:BY}
Let $S_l$ be the family of conformally flat lightcone cuts shrinking towards $p$ defined with respect to $(p,e_0)$ in an $n$-dimensional spacetime, the limits of the Brown-York energy as $l$ goes to $0$ are
\begin{enumerate}
\item In non-vacuum,
\beq
\lim_{l\rightarrow 0}l^{-(n-1)}M_{BY} (S_l) = \frac{\Omega_{n-2}}{n-1}T(e_0,e_0).\label{bynonvac}
\eeq
\item In vacuum or the stress tensor $T$ vanishes in an open set containing $p$,
\beq
\lim_{l\rightarrow 0}l^{-(n+1)}M_{BY} (S_l) = \Q .\label{byvac}
\eeq
\end{enumerate}
where the tensors $T,E,H,D$ are evaluated at $p$.
\end{thm}
\begin{rem}
The conformally flat lightcone cuts are a rather restrictive assumption we make in order for the reference to be well-defined. If such lightcone cuts do not exist for any $(p,e_0)$, then the small sphere limit above cannot be evaluated. The same applies to the Kijowski-Epp-Liu-Yau energy below.
\end{rem}
As noted by BLY \cite{brown1999canonical}, it is rather surprising that in four dimensions, both the Hawking energy and the BY energy yield the same vacuum limit as the two proposals are constructed from two totally different approaches. Here we choose the same lightcone reference as BLY, and we see that their small sphere limit agree in all dimensions. Our results suggest that the generalised Bel-Robinson superenergy $Q$, though unique, does not retain its gravitational energy interpretation beyond four dimensional spacetime. Nevertheless, we are reluctant to refer the new vacuum limit (\ref{hawkingvac0},\ref{byvac}) as a new candidate for the gravitational superenergy because this quantity is not always positive. In particular, whenever the magnetic-magnetic component of the Weyl tensor $D$ dominates over the other two contributions under some choice of $(p,e_0)$, we have a negative vacuum limit. This confirms the fact that the Hawking energy and the BY energy are plagued with non-positivity issues. In higher dimensions, it is more serious that non-positivity is manifested even in small lightcone cuts for some choices of $(p,e_0)$, which is not the case for $n=4$.

The Kijowski-Epp-Liu-Yau energy is a refinement of the Brown-York proposal in terms of the positivity. In order for the KELY energy to be defined, one needs two conditions: 1. the intrinsic Ricci scalar on $S_l$ is positive, $\R(l)>0$; 2. the Mean Curvature vector is spacelike, $\langle K,K \rangle >0$. These conditions are guaranteed on lightcone cuts $S_l$ for sufficiently small $l$. Again, for the ligtcone reference to be defined, we need to assume conformal flatness.

\begin{thm}\label{thm:LY}
Let $S_l$ be the family of conformally flat lightcone cuts shrinking towards $p$ defined with respect to $(p,e_0)$ in an $n$-dimensional spacetime, the limits of the Kijowski-Epp-Liu-Yau energy as $l$ goes to $0$ are
\begin{enumerate}
\item In non-vacuum,
\beq
\lim_{l\rightarrow 0}l^{-(n-1)}M_{KELY} (S_l) = \frac{\Omega_{n-2}}{n-1}T(e_0,e_0).\label{hawkingnonvac}
\eeq
\item In vacuum or the stress tensor $T$ vanishes in an open set containing $p$,
\beq
\lim_{l\rightarrow 0}l^{-(n+1)}M_{KELY} (S_l) = \Q-\frac{(n^2-n-3)E^2}{12(n^2-1)(n-2)^2(n-3)^2}.\label{hawkingvac}
\eeq
\end{enumerate}
where the tensors $T,E,H,D$ are evaluated at $p$.
\end{thm}
\begin{rem}
In four dimensions, the vacuum limit is the BR superenergy $Q_0$ with an extra term proportional to $E^2$. We see that in higher dimensions, the same pattern holds and $Q_0$ generalises to $\Q$ instead of $Q$ as for the Hawking energy and the BY energy. When $n=4$, the vacuum limit is positive but fails to remain positive in higher dimensions if $D^2$ dominates. 
\end{rem}

In section 2, we review the electromagnetic decompositions of the Weyl tensor and the Bel-Robinson superenergy in arbitrary dimensions; in section 3, the gauge freedom on the lightcone cuts is fixed and expansions of relevant geometric quantities are computed; in section 4, the generalisations of QLE's are defined and the assumptions are discussed; in section 5, we evaluate all the small sphere limits and prove the claimed theorems; and we finish by a short discussion in section 6.\\

 \emph{Notations.--}An $n$-dimensional spacetime is denoted as $(M^n,g)$ and the Minkowski spacetime is denoted as $(\M^n,\eta)$. Geometric quantities with a tilde $\tilde{ }$ live in $\M^n$, e.g. $\tilde{N_p}$ is a lightcone at $p$ in $\M^n$. We denote the Riemann curvature as $R_{abcd}$, the Ricci curvature as $Ric$ or $R_{ab}$ and Ricci scalar as $R$. The induced metric on a codimension-two surface $S$ is $h_{ab}$. $\R$ denotes the scalar curvature of $S$. $K^a$ denotes its mean curvature vector, $H$ denotes the mean curvature of $S$ as embedded in a $(n-1)$-hypersurface $\Sigma$ and $k$ denotes the mean curvature of $\Sigma$ as embedded in $(M^n,g)$. The outer and inner null generators on a spacelike closed $(n-2)$-surface are $\ell^+, \ell^-$ respectively. Their flat parts (to be precisely defined later) contracting with tensors are abbreviated as e.g. $R_{+-}=Ric(\l^+,\l^-).$  We use $a,b,c,\dots$ for abstract indices, $\mu,\nu,\alpha,\dots = 0,1,\dots,n-1$ for full-dimensional objects and $i,j,k,\dots = 1,\dots,n-1$ for codimension one objects in coordinates.  In Riemann Normal Coordinate (RNC)$\{x^\mu\}$ expansions, indices are raised or lowered by the Minkowski metric $\eta$. We do not distinguish upper/lower indices for contraction with respect to the Euclidean metric. (Square)brackets around indices indicate (anti-)symmetrisation.

\section{Review of the electromagnetic decomposition of the Weyl tensor and the Bel-Robinson superenergy}\label{sec:em}
We are interested in integral quantities on a lightcone cut $S_l$ in a perturbative regime. In non-vacuum, the curvature perturbations are characterised by Ricci-related quantities, like $R, R_{ab}, G_{ab}$. In vacuum, they can be organised by the electromagnetic decomposition of the Weyl tensor. Our discussion on the electromagnetic decomposition shall only include what we need. One can refer to \cite{senovilla2000super} for more details. 

\begin{defn}\label{emparts}
Given some timelike vector $e_0$ at $p$, in adapted coordinates $\{x^0,x^i\}$ where $\partial_{x^0}=e_0$, the Weyl tensor $C$ can be decomposed into spatial tensors
\beq
E_{ij} := C_{0i0j},\;\;\;\;\; H_{ijk} := C_{0ijk} ,\;\;\;\;\; D_{ijkl} := C_{ijkl},
\eeq
\end{defn}
where $E_{ij}$ is the electric-electric part, $H_{ijk}$ is the electric-magnetic part and $D_{ijkl}$ is the magnetic-magnetic part. 

In four dimensions, the Bel-Robinson tensor~\cite{Bel,robinson1997bel} is
\beq
Q^{(4)}_{abcd} = C_{aecf} C_b{}^e{}_d{}^f + C_{aedf} C_b{}^e{}_c{}^f - \frac{1}{2}g_{ab}C_{cefg}C\indices{_d^{efg}}
\eeq
which is defined in a way similar to how the electromagnetic stress tensor is built from the electromagnetic tensor. The BR tensor in four dimensions enjoys many nice properties, such as being traceless, totally symmetric and satisfying a conservation law \cite{senovilla2000super}. Most importantly, it satisfies the dominant property, which means that the tensor $Q^{(4)}_{abcd}$ contracted with any four future directed causal vectors is non-negative. The superenergy is defined as
\beq
Q_0:=Q^{(4)}(e_0,e_0,e_0,e_0)=E^2+B^2, \label{br4d}
\eeq
where $E^2:=E_{ij}E_{ij}, B_{ij}:=\frac{1}{2}\epsilon_{jkl}H\indices{_i_k_l}$. This form suggests the name `superenergy' analogous to the field energy in electrodynamics due to its different dimension. In fact, using dimensional analysis, one can argue that in four dimensional vacuum any Lorentz invariant quasilocal energy expression for a small sphere must be proportional to $Q_0$ at leading order \cite{szabados2009quasi}. This justifies the interpretation of $Q_0$ as purely gravitational energy. It is also a useful tool in studying the dynamics of general relativity such as the nonlinear stability of the Minkowski spacetime \cite{christodoulou2014global}.

By demanding the four-rank tensor being dominant and quadratic in the Weyl tensor, Senovilla discovered the following generalisation of the BR tensor in higher dimensions~\cite{senovilla2000super}.
\begin{align}
Q_{abcd} &= C_{aecf} C_b{}^e{}_d{}^f + C_{aedf} C_b{}^e{}_c{}^f 
- \frac{1}{2} g_{ab}  C_{gecf} C^{ge}{}_d{}^f-\frac{1}{2} g_{cd} C_{aegf} C_b{}^{egf}
+\frac{1}{8} g_{ab} g_{cd} C_{efgh}C^{efgh}. \label{eqn:Tarb}
\end{align}
$Q$ is not the unique tensor which satisfies above conditions when $n>5$, but the corresponding superenergy is. Namely,
\beq
Q:=Q(e_0,e_0,e_0,e_0)
\eeq
is an unique generalisation of the standard BR superenergy $Q_0$ in $n\geq 4$ given the tensor $Q_{abcd}$ is dominant and quadratic in Weyl. Using the Definition \ref{emparts}, one can rewrite $Q$ to be manifestly non-negative
\beq \label{eqn:W}
Q = \frac12\left[E^2+H^2+\frac14 D^2\right],
\eeq
where $E^2:=E_{ij}E_{ij}, H^2:=H_{ijk}H_{ijk}, D^2:=D_{ijkl}D_{ijkl}$. Note that $E_{ij}$ is basically the trace of $D_{ijkl}$ and $D_{ijkl}$ contain more information via the off diagonal entries for $n>4$. When $n=4$, they are equivalent and $D^2=4E^2, H^2=2B^2$, so $Q$ equals to $Q_0$. 

Even though $Q$ is unique, it may not necessarily acquire the physical meaning of a gravitational energy as $Q_0$ in four dimensions. Our results show that such an physical characterisation in terms of QLE is indeed missing here, where $\Q=\frac{(6n^2-20n+8)E^2+6(n-3)H^2-3D^2}{36(n-3)(n-2)(n^2-1)}$ replaces $Q$.\\

Before closing this section, we note some useful identities relating the electromagnetic parts,
\begin{lem}\label{lem:emidentity}
\begin{align}
D\indices{_{ik}_j_k}=E_{ij},\;\;\;H\indices{_k_l_i}H_{lki}=\frac{H^2}{2},\;\;\;H\indices{_k_i_l}H_{min}\delta^{(4)}_{klmn}=\frac{3}{2}H^2,\\
E_{kl}E_{mn}\delta^{(4)}_{klmn}=2E^2,\;\;\;\;\;D\indices{_p_k_i_l}D_{pmin}\delta^{(4)}_{klmn}=\frac{3}{2}D^2+E^2,
\end{align}
where $\delta^{(4)}_{klmn}:=\delta_{kl}\delta_{mn}+\delta_{km}\delta_{ln}+\delta_{kn}\delta_{lm}$.
\end{lem}
\begin{proof}
The first identity simply follows from
\beq
D\indices{_{ik}_j_k}=-C\indices{_{i0}_j^0}=C\indices{_{i0j0}}=E_{ij}.
\eeq
For the second identity, set $z$ as
\beq
z=H\indices{_k_l_i}H_{lki}=-H_{kli}(H_{kil}+H_{ilk})=H^2-H_{kli}H_{ilk}=H^2-H_{kil}H_{ikl}=H^2-z.\label{cal2}
\eeq
So we have $z=H^2/2$. Then consider
\begin{align}
H\indices{_k_i_l}H_{min}\delta^{(4)}_{klmn}=&H\indices{_k_i_l}H_{kil}+H\indices{_k_i_l}H_{lik}=H^2+\frac12H^2=\frac32H^2.
\end{align}
Also,
\beq
E_{kl}E_{mn}\delta^{(4)}_{klmn}=E_{kl}E_{kl}+E_{kl}E_{lk}=2E^2.
\eeq
Lastly, consider
\beq
D\indices{_p_k_i_l}D_{pmin}\delta^{(4)}_{klmn}=D_{pkik}D_{pnin}+D_{pkil}D_{pkil}+D_{pkil}D_{plik}=E^2+D^2+\frac12 D^2=E^2+\frac32 D^2.
\eeq
where $D_{pkil}D_{plik}=D^2/2$ follows from similar calculations as in (\ref{cal2}).
\end{proof}

\section{The geometry of small codimension-two submanifolds}
 In four dimensions, the standard method to evaluate the small sphere limit uses the Newman-Penrose formalism, or more generally the Geroch--Held--Penrose (GHP) formalism. Though we are aware of its generalisation to higher dimensions \cite{durkee2010generalization}, we choose to work with the standard tensorial framework of general relativity for greater accessibility. 

We use the Riemann Normal Coordinates $\{x^\mu\}$ \cite{brewin2009riemann} and choose the origin to be the lightcone vertex $p$. The metric on the lightcone cut $S_l$ reads
\begin{align}
g_{\al\be}(l) &= \eta_{\al\be}-\frac{l^2}3 R_{\al +\be +} 
-\frac{l^3}6 \nabla_+ R_{\al+\be+}+ l^4\left(\frac2{45} R\indices{^\ga_+_\al_+} R_{\ga +\be +}
-\frac1{20} \nabla_+ \nabla_+
R_{\al +\be +}\right) + O(l^5). 
\end{align}
where $+$ denotes contraction with $\l^+$ and all the curvature tensors are evaluated at $p$. For simplicity, we can omit some terms in the above expression that will be irrelevant after the integration on $S_l$. Terms involving covariant derivatives of the Riemann tensor will eventually vanish due to the following lemma 
\begin{lem}
\begin{align}
\delta_{ij}\nabla_0 C_{i0j0}=0,\;\;\; \delta^{(4)}_{ijkl}\nabla_k C_{ijl0}=0,\;\;\;\delta^{(4)}_{ijkl}\nabla_k \nabla_l C_{i0j0}=0,\;\;\;\delta_{ij}\delta^{(4)}_{klmn}\nabla_k \nabla_l C_{imjn}=0.
\end{align}
\end{lem}
\begin{rem}
In all the perturbative calculations we encounter later, the nontrivial terms involving derivatives of the Riemann tensor up to the order of curvature squared, which is the leading order in the vacuum case, reduce to the above forms after spherical integration. Similar arguments are also given in \cite{jacobson2018area}.
\end{rem}
\begin{proof}
The first identity vanishes as the Weyl tensor is traceless. The second identity vanishes as $i,j$ are antisymmetric in $C_{ijl0}$ and symmetric in $\delta^{(4)}_{ijkl}$.  Consider now
\begin{align}
\delta^{(4)}_{ijkl}\nabla_k \nabla_l C_{i0j0}=2\nabla_i\nabla_j C_{j0i0}=2\nabla_i\nabla_\al C\indices{^\al_{0i0}}=0
\end{align}
where the last equality is due to the Bianchi identity.  Similarly,
\begin{align}
\delta^{(4)}_{klmn}\nabla_k \nabla_l C_{imin}=\delta^{(4)}_{klmn}\nabla_k \nabla_l C_{m0n0}=2\nabla_k \nabla_l C_{l0k0}=0.
\end{align}
\end{proof}

We shall henceforth use the simplified version of the metric and leave out such derivative terms in all RNC expansions,
\beq
g_{\al\be}(l) = \eta_{\al\be}-\frac{l^2}3 R_{\al +\be +} + \frac{2l^4}{45} R\indices{^\ga_+_\al_+} R_{\ga +\be +} + O(l^5). \label{metric}
\eeq
and its inverse is
\beq
g^{\al\be}(l) = \eta^{\al\be}+\frac{l^2}3 R\indices{^\al_+^\be_+} +\frac{l^4}{15}R\indices{^\ga_+^\al_+} R\indices{_\ga_+^\be_+} + O(l^5).\label{inversemetric}
\eeq

The Levi-Civita connection is given by \cite{brewin2009riemann},
\beq\label{connection}
\Gamma^\mu_{\al\be}(l) = -\frac{2l}3R\indices{^\mu_{(\al\be)+}}+\frac{2l^3}{15}R\indices{_+^\ga_{+(\al}}R\indices{_{\be)}^\mu_+_\ga} +\frac{8l^3}{45}R\indices{_{(\al |+|\be)}^\ga}R\indices{_+^\mu_+_\ga} +O(l^4).
\eeq

The induced metric on $S_l$ reads
\beq \label{hmetric}
h_{ab} = g_{ab}(l) + \ell^+_a\ell^-_b + \ell^-_a\ell^+_b
\eeq
where $\ell^+, \ell^-$ are the outer (outgoing) and inner (ingoing) null generators on $S_l$. We shall refer to $\{\ell^+,\ell^-\}$ as the null frame.

The zeroth order contribution for $h$ is the Minkowski counterpart $\h$,
\beq
\h_{ab} = \eta_{ab}+\l^+_{a}\l^-_{b} + \l^-_{a}\l^+_{b}
\eeq
Note that by definition
\beq
h_{ab}\ell^{\pm b}=0, \;\;\h_{ab}\l^{\pm b}=0,
\eeq
and any objects orthogonal to $S_l$ will vanish under contractions with the induced metric.

We are interested in the QLE evaluated on the lightcone cuts, so it is natural to formulate all the quantities of interest in terms of the null frame variables defined on the lightcone, such as the expansion $\theta^\pm$ and the shear $\sigma^\pm_{ab}$. We should first fix the null frame $\{\ell^+,\ell^-\}$. For a small lightcone cut $S_l$, the leading contribution to the null generators in the RNC expansion are the their flat counterparts $\l^\pm$. By making a gauge choice for $\l^{\pm\mu}$ consistent with normalisations (\ref{normalisation1},\ref{normalisation2}), we can fix the exact expansions of the null tangents that are normalised and hypersurface orthogonal. 
\begin{prop}\label{lem:gauge}
Choose the leading contribution $\l^\pm_\mu$ to the outer and inner null generators as
\beq
\l^{+\mu} := (1, n^i),\;\; \l^{-\mu} :=\frac12(1,-n^i),
\eeq
the RNC expansions of $\ell^\pm_\mu$ restricted on $S_l$ are given by:
\beq
\begin{aligned}
\ell^{+\mu} =& (1, n^i) = \l^{+\mu}, \\
\ell^+_\mu =& (-1, n^i)  = \l^+_\mu :=\eta_{\mu\nu}\l^{+\nu}, \\
\ell^{-\mu} =& \l^{-\mu}+\frac{l^2}{6}R_{+-+-}\ell^{+\mu}+\frac{l^2}{3}R\indices{_{+-+}^\mu}+\frac{l^4}{30}R\indices{^\nu_{+-+}}R_{\nu +-+}\ell^{+\mu}+\frac{l^4}{15}R\indices{^\nu_{+-+}}R\indices{_{\nu +}^\mu_+}+O(l^5),\\
\ell^{-}_\mu =& \l_{\mu}^-+\frac{l^2}{6}R_{+-+-}\ell^+_{\mu}+\frac{l^4}{30}R\indices{^\nu_{+-+}}R_{\nu +-+}\ell^+_{\mu}+O(l^5),
\end{aligned}
\eeq
where $n^i$ is a normalised spacelike vector indicating the spatial direction, $\ell^\pm_{\mu}=g_{\mu\nu}\ell^{\pm\nu}, \l^\pm_{\mu}=\eta_{\mu\nu}\l^{\pm\nu},$ and we use abbreviations such as $R_{+-+-}=R(\l^+,\l^-,\l^+,\l^-),$ etc.
\end{prop}
\begin{rem}
Here we will only deal with the geometric quantities restricted to $S_l$, so we do not need to specify the extensions of $\ell^\pm$ beyond $S_l$. 
\end{rem}

\begin{proof}
Via a direct calculation, we see that the generators $\l^\pm$ in $\M^n$ satisfy
\beq
\partial_\mu \l^+_{\nu}=\frac1{l}\h_{\mu\nu},\;\;\;\partial_\mu \l^-_{\nu}=-\frac{1}{2l}\h_{\mu\nu}\label{partialell}
\eeq
We would like to impose the normalisaitons $ \langle\ell^+,\ell^-\rangle =-1, \langle \ell^\pm,\ell^\pm\rangle=0\; \forall l,$ and twist-freeness $\omega_{ab}:=h^\al_{[\mu}h^\be_{\nu]}\nabla_\al\ell^\pm_\be=0$, which then implies hypersurface orthogonality. 
For the outer null tangent, consider the general form of the expansion,
\beq
\ell^{+\mu} = \l^{+\mu} + A^\mu(R) l^2 + B^\mu (R^2) l^4 + O(l^6)
\eeq
where again we leave out irrelevant terms containing curvature derivatives for simplicity.

Now we impose
\beq
\langle \ell^+,\ell^+\rangle = 0.
\eeq
Using the expansion of the metric (\ref{metric})
\beq
\langle \l^+, \l^+\rangle_\eta +2\langle \l^+,A\rangle_\eta l^2 +2\langle \l^+,B\rangle_\eta l^4 = 0,
\eeq
which has to vanish order by order. Our choice $\l^{+\mu} = (1, n^i) $ satisfies  $\langle \l^+, \l^+\rangle_\eta=0$, and $A=B=0$. 
We then obtain
\beq
\ell^{+\mu}=\l^{+\mu}=(1,n^i),\ell^+_\mu=(-1,n^i).
\eeq
One can easily check that the twist vanishes for $\ell^+_\mu$ using (\ref{partialell}). 

Similarly, consider the general form of the inner null tangent,
\beq
\ell^{-}_\mu = \l^{-}_\mu + A_\mu(R) l^2 + B_\mu (R^2) l^4 + O(l^6)
\eeq
where we choose to work with the 1-form for convenience.

The normalisation conditions implies
\begin{align}
A_+ l^2 + B_+ l^4=0,\;\;\;\; 2A_- l^2 +  \frac{l^2}{3}R\indices{_-_{+-+}}+\frac{2l^4}{3}R\indices{^\mu_{+-+}}A_\mu + \frac{2l^4}{15}R\indices{^\al_+_-_+}R\indices{_\al_+_-_+} + 2B_- l^4=0
\end{align}
which imply
\begin{align}
A_+=B_+=0, 2A_-+\frac{1}3 R_{+-+-} =0, \label{A+A-}\\   
2B_-+\frac{1}{15}R\indices{^\al_+_-_+}R\indices{_\al_+_-_+}+\frac{2}{3}R\indices{^\mu_{+-+}}A_\mu  =0.\label{B+B-}
\end{align}
The first line (\ref{A+A-}) imply 
\beq
A_\mu=\frac16(c_1 R_{+-+-}\ell^+_\mu - c_2 R_{+-+\mu})
\eeq
for some real coefficients with $c_1+c_2=1$.

Now consider the twist,
\beq
\begin{aligned}
h^\al_{[\mu}h^\be_{\nu]}\nabla_\al\ell^-_\be = h^\al_{[\mu}h^\be_{\nu]}\partial_\al\ell^-_\be=&h^\al_{[\mu}h^\be_{\nu]}\partial_\al\left(\l^-_\be+\frac{l^2}6(c_1 R_{+-+-}\ell^+_\be - c_2 R_{+-+\be})+O(l^4)\right),\\
=&\frac{l^2}{6}h^\al_{[\mu}h^\be_{\nu]}(c_1R_{+-+-}\partial_\al\ell^+_\be-c_2\partial_\al R_{+-+\be}),\\
=&-c_2\frac{l^2}{6}h^\al_{[\mu}h^\be_{\nu]}\partial_\al R_{+-+\be}
\end{aligned}
\eeq
where we use (\ref{partialell}) and in the second line, the partial derivative $\partial_\al R_{+-+-}\ell^+_\be$ contains other terms, but since $h^\be_\nu\ell^+_\be = 0$, only $R_{+-+-}\partial_\al\ell^+_\be$ contributes. 

The twist does not vanish in general unless $c_2=0$. So we have
\beq
A_\mu=\frac16 R_{+-+-}\ell^+_\mu,
\eeq
and together with (\ref{B+B-}) it implies
\beq
B_\mu=\frac{l^4}{30}(d_1R\indices{^\nu_{+-+}}R_{\nu +-+}\ell^+_{\mu}-d_2R\indices{^\nu_{+-+}}R_{\nu +\mu+})
\eeq
and again we set $d_2$ to zero as well for the twist to vanish.

Altogether, we obtain 
\beq
\ell^{-}_\mu =\l_{\mu}^-+\frac{l^2}{6}R_{+-+-}\ell^+_{\mu}+\frac{l^4}{30}R\indices{^\nu_{+-+}}R_{\nu +-+}\ell^+_{\mu}+O(l^5),
\eeq
Using the inverse metric (\ref{inversemetric}) we obtain the expression for $\ell^{-\mu}$.
\end{proof}

After fixing the gauge, we can work with null frame variables, namely the expansion and the shear, 
\begin{defn}
The expansion $\theta^\pm$ and the shear $\sigma^\pm_{ab}$ associated with null generators $\ell^\pm$ are defined as
\begin{align}\label{nulldefintions}
\theta^\pm &:=h^{ab}\nabla_a\ell^\pm_b,\\
\sigma^\pm_{ab} &:= h^c_{(a}h^d_{b)}\nabla_c\ell^\pm_d - \frac{1}{n-2}\theta^\pm h_{ab}.
\end{align}
\end{defn}

\begin{lem}\label{lem:gauss}
The contracted Gauss equation in terms of the null frame variables reads
\begin{align}
\frac{\R}2+\frac{n-3}{n-2}\theta^-\theta^+ - \sigma^-_{ab}\sigma^{+ab} = \frac{R}2+2Ric(\ell^+,\ell^-)-R(\ell^+,\ell^-,\ell^+,\ell^-).\label{eqn:gauss}
\end{align}
\end{lem}
\begin{proof}
The contracted Gauss equation constraining the geometry of the codimension-2 surface is given by
\beq
\R-K^{+ab}K^-_{ab} -K^aK_a = h^{ac}h^{bd}R_{abcd}. \label{Gauss0}
\eeq
where $K^\pm_{ab}$ are the second fundamental forms with respect to $\ell^\pm$ and $K^a$ is the mean curvature vector. Our convention here follows \cite{ge2018memorial}.

Using (\ref{hmetric}) we have,
\beq
 h^{ac}h^{bd}R_{abcd}=R+4Ric(\ell^+,\ell^-)-2R(\ell^+,\ell^-,\ell^+,\ell^-).\label{ee1}
\eeq
The second fundamental forms are related to the null frame variables by 
\begin{align}
K^\pm_{ab} = \frac{1}{n-2}h_{ab}\theta^\pm + \sigma^\pm_{ab},\;\;\;\;K^\pm:=\text{tr} K^\pm_{ab} =-\theta^\pm = K^a\ell^\pm_a
\end{align}

\beq
\langle K, K \rangle_g = \langle K,K \rangle_h - \ell^+_aK^a \ell^-_b K^b -  \ell^-_aK^a \ell^+_b K^b= -2\theta^+\theta^- \label{ee2}
\eeq
where $\langle K,K \rangle_h=0$ by definition.
\beq
K^{+ab}K^-_{ab}=\frac{\theta^+\theta^-}{n-2}+\sigma^+_{ab}\sigma^{-ab} \label{ee3}
\eeq
Substitute (\ref{ee1},\ref{ee2},\ref{ee3}) into (\ref{Gauss0}) yields the claimed Gauss equation (\ref{eqn:gauss}).
\end{proof}

\begin{lem}\label{lem:outerexp}
The outer expansion and shear on the lightcone cut $S_l$ are
\begin{align}
\theta^+(l)& = \frac{n-2}{l} -\frac{l}{3}R_{+-}- \frac{C_{\mu+\nu+}C\indices{^\mu_+^\nu_+}}{45} l^3+P(Ric) l^3 +O(l^4),\label{eqn:expansion}\\
\sigma_{\mu\nu}^+(l)& = -\frac{l}{3}h^\al_\mu h^\be_\nu C_{\al+\be+}+ P(Ric) l^3
+O(l^4),\label{shearexp}
\end{align} 
\end{lem}
Note that in order to keep the expressions clean, we have packaged terms that contain Polynomials of the Ricci curvature at order $O(l^3)$ into $P(Ric)$, because contributions at this order will only be relevant for our calculations in the vacuum case and $P(Ric=0)=0$.
\begin{proof}
We can directly apply Proposition \ref{lem:gauge} and compute the outer expansion and shear using the definition but this cumbersome calculation can be avoided by solving the coupled evolution equations on the lightcone.
\begin{align}
\dot{\theta^+} &= -\frac{1}{n-2}\theta^{+2} - \sigma^{+\mu\nu} \sigma^+_{\mu\nu} - Ric(\ell^+,\ell^+),\label{eqn:raychaudhuri1}\\
\dot{\sigma^+}_{\mu\nu} &= -\frac{2}{n-2}\theta^+\sigma^+_{ab}-h^\al_\mu h^\be_\nu C_{\al\ga\be\lambda}\ell^{+\ga}\ell^{+\lambda} ,\label{eqn:raychaudhuri2}
\end{align}
where the first is the Raychaudhuri equation and the second is the evolution equation of the shear.

The ODEs can be perturbatively solved by a power series ansatz. The leading order is given by $\theta^+=(n-2)/l$ in Minkowski spacetime. Plugging it into (\ref{eqn:raychaudhuri2}) yields $\sigma^+_{\mu\nu}=-C_{\al+\be+} h^\al_\mu h^\be_\nu l/3$. Therefore, we use the following ansatz:
\begin{align}
\theta^+(l)& = \frac{n-2}{l} + c_0 + c_1 l + c_2 l^2 + c_3 l^3 +O(l^4),\\
\sigma^+_{\mu\nu}(l)& = -\frac{l}{3}C_{\al+\be+} h^\al_\mu h^\be_\nu + k_2 l^2 + k_3 l^3 +O(l^4).
\end{align}
Solving (\ref{eqn:raychaudhuri1},\ref{eqn:raychaudhuri2}) simultaneously yields:
\begin{align}\label{eqn:expansion}
\theta^+(l)& = \frac{n-2}{l} -\frac{Ric(\ell^+,\ell^+)}{3} l - \frac{(n-2)C_{\mu+\nu}C\indices{^\mu_+^\nu_+}+(Ric(\ell^+,\ell^+))^2}{45(n-2)} l^3+O(l^4),\\
\sigma^+_{\mu\nu}(l)& = -\frac{l}{3}C_{\al+\be+} h^\al_\mu h^\be_\nu - \frac{2 C_{\al+\be+} h^\al_\mu h^\be_\nu Ric(\ell^+,\ell^+)}{45(n-2)} l^3+O(l^4).
\end{align}
Packaging the $Ric$-related terms at $O(l^3)$ yields the result. 
\end{proof}

\begin{lem}\label{lem:innerexp}
The inner expansion on the lightcone cut $S_l$ is
\beq
\begin{aligned}
\theta^-(l)=&-\frac{n-2}{2l}-\left(\frac23 R_{+-}+\frac16 R_{++}-\frac{n+2}{6}R_{+-+-}\right)l\\
&+\left(\frac{n+6}{30}C\indices{^\mu_{+-+}}C_{\mu +-+}-\frac4{45}C\indices{^\mu_+^\nu_+}C_{\mu +\nu -}-\frac{1}{30}C\indices{^\mu_+^\nu_+}C_{\mu +\nu +}\right)l^3 + P(Ric) l^3+ O(l^4)
\end{aligned} 
\eeq
\end{lem}
in which again we have packaged irrelevant contributions that is polynomial in Ricci curvature.
\begin{proof}
Plugging in the expression of $\ell^-_\mu$ in Proposition \ref{lem:gauge} to Definition \ref{nulldefintions}, and using (\ref{metric},\ref{inversemetric},\ref{connection}) we have
\beq
\begin{aligned}
\theta^- =& h^{\mu\nu}\partial_\mu\ell^-_{\nu} -h^{\mu\nu}\Gamma^\al_{\mu\nu}\ell_\al^-,\\
=&g^{\mu\nu}\left(-\frac1{2l}\h_{\mu\nu}+\frac{l}{6}R_{+-+-}\h_{\mu\nu}+\frac{l^3}{30}R\indices{^\al_{+-+}}R_{\al +-+}\h_{\mu\nu}\right)\\
&-h^{\mu\nu}\left(\l_{\al}^-+\frac{l^2}{6}R_{+-+-}\ell^+_{\al}\right)\left( -\frac{2l}3R\indices{^\al_{\mu\nu+}}+\frac{2l^3}{15}R\indices{_+^\ga_{+\mu}}R\indices{_{\nu}^\al_+_\ga} +\frac{8l^3}{45}R\indices{_{\mu+\nu}^\ga}R\indices{_+^\al_+_\ga}\right)+O(l^4),\\
=&-\frac{n-2}{2l}-\frac{l}{6}R_{++}-\frac{l^3}{30}C_{\mu+\nu+}C\indices{^\mu_+^\nu_+} + \frac{n-2}{6}R_{+-+-}l+\frac{n-2}{30}C\indices{^\al_{+-+}}C_{\al +-+}l^3\\
&-\frac{2l}3R_{+-}+\frac{2l}3R_{+-+-}+\frac{2l^3}{9}(C_{\mu+-+}C\indices{^\mu_{+-+}}-C_{\mu+\nu+}C\indices{^\mu_+^\nu_-}) \\
&+\frac{2l^3}{15}(C\indices{_+^\ga_{+\mu}}C\indices{_-_{\mu}_+_\ga}-C_{\mu+-+}C\indices{^\mu_{+-+}}) +\frac{8l^3}{45}C\indices{_{+-+}^\ga}C\indices{_+_-_+_\ga}+P(Ric)l^3+O(l^4)\\
=&-\frac{n-2}{2l}-\left(\frac23 R_{+-}+\frac16 R_{++}-\frac{n+2}{6}R_{+-+-}\right)l\\
&+\left(\frac{n+6}{30}C\indices{^\mu_{+-+}}C_{\mu +-+}-\frac4{45}C\indices{^\mu_+^\nu_+}C_{\mu +\nu -}-\frac{1}{30}C\indices{^\mu_+^\nu_+}C_{\mu +\nu +}\right)l^3 + P(Ric) l^3+ O(l^4)
\end{aligned}
\eeq
where in the second line we used the fact that $h^{\mu\nu}\h_{\mu\nu}=g^{\mu\nu}\h_{\mu\nu}$ and we expanded $h^{\mu\nu}$ to obtain the third equality.
\end{proof}
Regarding the inner shear, we can save the cumbersome computation of $\sigma^-_{\mu\nu}$ as only the combination $\sigma^-_{\mu\nu}\sigma^{+\mu\nu}$ will be needed. We shall also computed $\theta^-\theta^+$ which will be used later. 
\begin{lem}\label{lem:invariants}
The gauge-invariant products of expansion and shear on $S_l$ are  
\begin{align}
\sigma^-_{\mu\nu}\sigma^{+\mu\nu}(l) =&
\frac{l^2}{18}\left(C_{\mu+\nu+}C\indices{^\mu_+^\nu_+}+4C_{\mu+\nu-}C\indices{^\mu_+^\nu_+}-4C_{\mu+-+}C\indices{^\mu_+_-_+}\right)+P(Ric) l^2+ O(l^3).\\
\theta^-\theta^+(l)=&-\frac{(n-2)^2}{2l^2}-\frac{n-2}{6}\left[4R_{+-} - (n+2)R_{+-+-}\right]+P(Ric)l^2\nonumber\\
&+ \frac{l^2(n-2)}{15}\left(\frac{n+6}{2}C\indices{^\mu_{+-+}}C_{\mu +-+}-\frac4{3}C\indices{^\mu_+^\nu_+}C_{\mu +\nu -}-\frac{1}{3}C\indices{^\mu_+^\nu_+}C_{\mu +\nu +}\right)+O(l^3).
\end{align} 
\end{lem}

\begin{rem}
$\theta^-\theta^+$ is invariant and it is related to the mean curvature vector $K$ and the Newman-Penrose variables $\rho,\mu$ as
\begin{align}
\theta^-\theta^+ = -\frac12\langle K,K \rangle =-4\rho\mu .
\end{align}
\end{rem}
\begin{proof}
The products of the shear can be simplified. Using the Definition \ref{nulldefintions}, we have
\beq
\begin{aligned}
\sigma^-_{\mu\nu}\sigma^{+\mu\nu}(l) =&\left(h^\al_{(\mu}h^\be_{\nu)}\nabla_\al\ell^-_\be - \frac{1}{n-2}\theta^- h_{\mu\nu}\right) \sigma^{+\mu\nu} = \sigma^{+\mu\nu}\nabla_\mu\ell^-_\nu ,\\
=&-\frac{l}{3}h^{\mu}_{\mu '}h^{\nu}_{\nu '}C\indices{^{\mu '}_+^{\nu '}_+}\left(\partial_\mu\ell^-_{0\nu}+\partial_\mu \frac{l^2}{6}C_{+-+-}\ell^+_\nu - \frac{2l}{3}C_{-\mu+\nu}\right)+P(Ric)l^2+O(l^3),\\
=&-\frac{l}{3}h^{\mu}_{\mu '}h^{\nu}_{\nu '}C\indices{^{\mu '}_+^{\nu '}_+}\left(-\frac{1}{2l}h^0_{\mu\nu}+\frac{l}{6}C_{+-+-}h^0_{\mu\nu} - \frac{2l}{3}C_{-\mu+\nu}\right)+P(Ric)l^2+O(l^3)
\end{aligned}
\eeq
where in the second line, we used (\ref{connection}). To proceed, we need the following identity:
\begin{align}
h^{\mu}_{\mu '}h^{\nu}_{\nu '}h^0_{\mu\nu}=h^{\mu}_{\mu '}h^{\nu}_{\nu '}h_{\mu\nu}+h^{\mu}_{\mu '}h^{\nu}_{\nu '}(h^0_{\mu\nu}-h_{\mu\nu})=h_{\mu'\nu'}+\frac{l^2}{3}(R_{\mu'+\nu'+}-R_{+-+-}\ell^+_{\mu'}\ell^+_{\nu'})
\end{align}
where we use Proposition \ref{lem:gauge} and Lemma (\ref{metric}). 

So we have
\beq
\begin{aligned}
\sigma^-_{\mu\nu}\sigma^{+\mu\nu}(l)=&\frac{l^2}{18}C_{\mu+\nu+} C\indices{^{\mu}_+^{\nu}_+}+ \frac{2l^2}{9}h^{\mu}_{\mu '}h^{\nu}_{\nu '}C\indices{^{\mu '}_+^{\nu'}_+}C_{-\mu+\nu}+P(Ric)l^2+O(l^3),\\
=&\frac{l^2}{18}C_{\mu+\nu+} C\indices{^{\mu}_+^{\nu}_+}+ \frac{2l^2}{9}C\indices{^{\mu}_+^{\nu}_+}C_{-\mu+\nu}- \frac{2l^2}{9}C_{+-+\mu}C\indices{_{+-+}^\mu}+P(Ric)l^2+O(l^3).
\end{aligned}
\eeq
The product of the expansions simply follows from Lemma \ref{lem:outerexp},\ref{lem:innerexp}. 
\end{proof}
\begin{lem} \label{lem:ricci}
The Ricci scalar of $S_l$ is
\beq
\begin{aligned}
\R(l) =& \frac{(n-2)(n-3)}{l^2}+R+\frac{4n}{3}R_{+-}-\frac{n(n-1)}{3}R_{+-+-}+P(Ric) l^2 \\
&+\frac{l^2}{45} \left((4 C_{\mu+\nu-}C\indices{^\mu_+^\nu_+} + C_{\mu+\nu+}C\indices{^\mu_+^\nu_+}) (2 n-1) + (4 - 9 n - 3 n^2)C\indices{^\mu_{+-+}}C_{\mu +-+}\right)+O(l^3).
\end{aligned}
\eeq
\end{lem}
\begin{proof}
It simply follows from the Gauss equation, and Lemma \ref{lem:gauss},\ref{lem:invariants}.
\end{proof}

From Lemma \ref{lem:outerexp}, one can immediately compute the area of the lightcone cut.
\begin{prop}{\cite{wang2019geometry}}\label{cutarea}
The area of a small lightcone cut $S_l$ in $M^n$ is 
\beq
A(S_l) = A^\flat\left(1-l^4\frac{2(n^2+2)E^2+3D^2+6nH^2}{360(n^2-1)}\right)+O(l^{n+3}),
\eeq
where $A^\flat:=\Omega_{n-2}l^{n-2}$ is the area in the Minkowski spacetime.
\end{prop}
\begin{proof}
The expansion governs the rate of change of the area along the null congruence. The induced volume form on $S_l$ satisfies:
\beq\label{eqn:areachange}
\dot{\sqrt{h}} = \theta^+\sqrt{h},
\eeq
where the dot represents the derivative with respect to the affine parameter $l$ of the null generators.

We use a perturbative ansatz
\beq
\sqrt{h}(l) = \Omega_{n-2}l^{n-2}\left(1+q_1 l + q_2 l^2 + q_3 l^3 + q_4 l^4 \right) + O(l^{n+3}).
\eeq
Plugging in the ansatz and (\ref{eqn:expansion}) into (\ref{eqn:areachange}) and set $Ric=0$ yields:
\begin{align}
\sqrt{h}(l) = \Omega_{n-2}l^{n-2}- \Omega_{n-2}l^{n+2}\frac{C_{\mu+\nu+}C\indices{^\mu_+^\nu_+}}{180}+O(l^{n+3}). \label{eqn:areaelement}
\end{align}

Using Lemma \ref{lem:identities}, the area is
\begin{align}
A(S_l)= \int_{S_l} \sqrt{q}\;\dd \Omega_{n-2}&=\Omega_{n-2}l^{n-2}\left(1-l^4\frac{2(n^2+2)E^2+3D^2+6nH^2}{360(n^2-1)}\right)+O(l^{n+3}).
\end{align}
\end{proof}
In four dimensions, the vacuum small sphere limits of various QLE's and the area are all dominated by the contribution $\Psi^0_0\overline{\Psi_0^0}$ (in terms of Newman-Penrose spin coefficients) which is proportional to the BR superenergy $Q$. One can thus regard this as the source of $Q$. Proposition \ref{cutarea} gives us hints that the same no longer holds at higher dimensions.\\

The following integral identities will be useful later.
\begin{lem}\label{lem:identities}
\beq
\begin{aligned}
\int_{S^{n-2}} \dd \Omega_{n-2} n^in^j =& \frac{\Omega_{n-2}}{n-1}\delta_{ij},\\
\int_{S^{n-2}}  \dd \Omega_{n-2} n^in^jn^kn^l =& \frac{\Omega_{n-2}}{n^2-1}(\delta_{ij}\delta_{kl}+\delta_{ik}\delta_{jl}+\delta_{il}\delta_{jk}),\\
\int_{S^{n-2}}  \dd \Omega_{n-2} R_{+-+-} =& \frac{\Omega_{n-2}}{n-1}Ric(e_0,e_0),\\
\int_{S^{n-2}}  \dd \Omega_{n-2} C_{+-+-} =& 0,\\
\int_{S^{n-2}}  \dd \Omega_{n-2} R_{++} =& \Omega_{n-2}\frac{nRic(e_0,e_0)+R}{n-1},\\
\int_{S^{n-2}}  \dd \Omega_{n-2} R_{+-} =& \Omega_{n-2}\frac{(n-2)Ric(e_0,e_0)-R}{2(n-1)},\\
\int_{S^{n-2}}  \dd \Omega_{n-2} C_{\mu+\nu+}C\indices{^\mu_+^\nu_+} =& \Omega_{n-2}\left(\frac{2(n^2+2)E^2+6nH^2+3D^2}{2(n^2-1)}\right),\\
\int_{S^{n-2}}  \dd \Omega_{n-2} C_{\mu+\nu-}C\indices{^\mu_+^\nu_+} =& \Omega_{n-2}\left(\frac{2(n^2-4)E^2+6H^2-3D^2}{4(n^2-1)}\right),\\
\int_{S^{n-2}}  \dd \Omega_{n-2} C_{\mu+-+}C\indices{^\mu_+_-_+} =& \Omega_{n-2}\left(\frac{2(n-1)E^2+3H^2}{2(n^2-1)}\right),\\
\int_{S^{n-2}}  \dd \Omega_{n-2} C_{+-+-}^2 =&  \Omega_{n-2}\frac{2E^2}{n^2-1}
\end{aligned}
\eeq
where $\dd\Omega_{n-2}$ is the natural volume form on a unit $S^{n-2}$ and $\Omega_{n-2}:=\int_{S^{n-2}}\dd\Omega_{n-2}$ and $n^i$ are normalised spacelike vectors. 
\end{lem}
We leave the proof of the above Lemma in \ref{app:proof1}.

\section{The definitions of quasilocal energy in higher dimensions}\label{sec:def}
In higher dimensions, the Einstein field equation without the cosmological constant reads
\beq
R_{ab}-\frac12 R g_{ab} = \Omega_{n-2}(n-2)G_N T_{ab}\label{EFE}
\eeq
where $G_N $ is the Newton's constant. We use the convention that instead of keeping $8\pi$, the factor $\Omega_{n-2}(n-2)$ is pulled out explicitly in accordance with the Green's function of the laplacian in $\mathbf{R}^{n-1}$. One is of course free to use the standard conventions with $8\pi$ but then the following QLE definitions should change accordingly. We henceforth set $G_N=1$ for convenience.

The standard Hawking energy \cite{hawking1968gravitational} in four dimensions is defined as
\beq\label{Hawking0}
M_{0}(S) = \sqrt{\frac{\text{Vol}(S)}{16\pi}}\left(1- \frac1{16\pi}\int_S H^2 \dd\sigma\right) 
= \sqrt{\frac{\text{Vol}(S)}{16\pi}}\left(1+ \frac{1}{8\pi}\int_S \theta^-\theta^+ \dd \sigma \right)
\eeq
where Vol(S) refers to the area of the 2-surface. The two expressions are equivalent. The Hawking energy fails to satisfy the rigidity condition of QLE: $M_0$ does not vanish for some 2-surfaces in the Minkowski spacetime. This problem is remedied by the Hayward energy \cite{hayward1994quasilocal},
\beq\label{Hayward2}
M_{1}(S) = \sqrt{\frac{\text{Vol}(S)}{16\pi}}\left(1+ \frac{1}{8\pi}\int_S\left[ \theta^-\theta^+ - 2\sigma^-_{ab}\sigma^{+ab} \right]\dd\sigma \right).
\eeq
In higher dimensions, we study the following generalisation of the HH energy,
\begin{defn}
For a codimension-2 closed spacelike surface $S$ in an $n$-dimensional spacetime, the Hawking-Hayward type energy is defined as
\beq\label{Hawking}
M_\alpha(S) = \frac{\left(\frac{\text{Vol}(S)}{\Omega_{n-2}}\right)^{\frac{1}{n-2}}}{(n-2)(n-3)\Omega_{n-2}}\int_S \left( \frac{\R}{2} + \frac{n-3}{n-2} \theta^-\theta^+ -\alpha \sigma^-_{ab}\sigma^{+ab} \right)\dd\sigma
\eeq
where $\sigma$ is the induced volume form on $S,\;\text{Vol}(S)=\int_S\dd\sigma$, and  $\alpha\in\mathbf{R}$ parameterises the Hayward modification. 
\end{defn}
\begin{rem}
The same generalisation of the Hawking energy $M_0$ has been studied in \cite{miao2017quasi} and appears in a discussion of quasilocal energy in \cite{bousso2019outer}. The original definition by Hayward \cite{hayward1994quasilocal} also contains an anholonomicity term, but it is gauge dependent (See discussions in sections 6.3 and 4.1.8 in \cite{szabados2009quasi}). Therefore, we use the version as defined in section 6.3 in \cite{szabados2009quasi}.
\end{rem}
We keep $\alpha$ flexible for our generalisation following \cite{bergqvist1994energy}. One can readily check that when $n=4$, $M_0$ and $M_1$ reduce to (\ref{Hawking0},\ref{Hayward2}) respectively via the Gauss-Bonnet theorem.

This is a natural generalisation because it retains the properties that the original Hawking energy satisfies (One can find a list of criteria that a sound QLE proposal should comply with in \cite{ChristodoulouYau,szabados2009quasi,wang2015four}). More specifically, $M_0$ reduces to the Misner-Sharp energy (See \ref{app:MS} for the definition of the Misner-Sharp energy in $n$ dimensions) for round spheres in spherically symmetric spacetime and has been shown to yield the ADM mass at spatial infinity \cite{miao2017quasi}. By requiring these properties, and our result of the non-vacuum limit in Theorem \ref{thm:hawking}, the coefficients in the generalisation $M_\alpha$ are uniquely fixed up to $\alpha$. \\

Unlike the Hawking-Hayward energy, the original definitions of the BY energy and KELY energy can be directly carried over to higher dimensions besides the subtleties about the zero energy references. In four dimensions, BLY and Yu uses the lightcone reference, in which $S_l$ is isometrically embedded to a lightcone in $\M^4$. For our purposes, we shall defined the BY energy and KELY energy in all dimensions with respect to the ligthcone reference as well so that our results are comparable with earlier works, and we need to assume the existence of the isometric embedding by imposing conformal flatness as discussed earlier.

The Brown-York energy is not a covariant proposal for QLE,  it relies on the codimension-two surface being defined on a hypersurface. We choose a particular family of hypersurfaces $\Sigma$ by fixing the normal vector $u$, such that when constrained on a lightcone cut $S:=\Sigma \cap N_p,$ it is given by
\beq
u := \frac12 \ell^+ +\ell^- ,
\eeq
and the normal of $S_l$ in $\Sigma$ is
\beq
v := \frac12\ell^+-\ell^-,
\eeq
and $\langle u,u \rangle=-1,\;\;\langle v,v \rangle=1$ and $\langle u,v \rangle = 0$.

This choice follows exactly from BLY \cite{brown1999canonical}. Therefore the mean curvature of $S_l$ as embedded in $\Sigma$ is
\beq
H := -\theta^-+\frac{\theta^+}{2}.
\eeq 
More precisely, the components of the mean curvature vector $K=K^1 v + K^0 u$ read
\begin{align}
K^1 = H = -2\mu-\rho = \frac{\theta^+}{2}-\theta^-, \;\;\;K^0 = k = 2\mu-\rho = \frac{\theta^+}{2}+\theta^-,
\end{align}
so 
\beq
\langle K,K \rangle = H^2 - k^2 = -2\theta^+\theta^-. 
\eeq
After fixing the hypersurface, the small sphere limit hinges on the choice of lightcone reference. We need to do the same in the Minkowski reference. By assuming conformal flatness, we know there exists an isometrically embedded surface $\tilde{S}$ on a Minkowski lightcone $\tilde{N}_p$. So we need to impose conditions on $\tilde{\Sigma}$ that intersects $\tilde{N}_p$ at $\tilde{S}.$ In four dimensions,  BLY requires the outer expansions being identical. More precisely, $\tilde{\Sigma}$ satisfies: 
\begin{enumerate}
\item[$\star$] The outer expansion $\tilde{\theta^+} := \tilde{H} + \tilde{k} = \theta^+$, where $\tilde{H} (\tilde{k})$ is the mean curvature with respect to the normalised normal $v (u)$ of $\tilde{\Sigma} (\tilde{S})$ in $\M^n (\tilde{\Sigma})$. 
\end{enumerate}
Above we denote all the reference space counterparts with $\tilde{ }$ . Alternatively, one can use the Euclidean reference where one embeds $S$ to $\mathbf{R}^{n}$. However, it is shown by BLY that in four dimensions, the limit deviates from the BR superenergy $Q_0$. We believe it is perhaps a less physical choice as compared to the lightcone reference in this context. Therefore, in higher dimensions, we also would like to make the same embedding configuration for consistency.

With condition $\star$, the mean curvature $\tilde{H}$ satisfies the vacuum Gauss equation (Lemma \ref{lem:gauss}):
\beq
\tilde{H}^2-(\theta^+ - \tilde{H})^2=\frac{n-2}{n-3}\R \label{vacuumgauss}
\eeq
where the shear vanishes as $\tilde{S}$ sits on a Minkowski lightcone. We see that our choice of the reference mean curvature $\tilde{H}$ is quasilocal, i.e. it depends only on the data on $S_l$. (\ref{vacuumgauss}) also means the embedded surface satisfying $\star$ has a unique $\tilde{H}$ value under condition $\star$ above. Note this is not at odds with the rigidity property of Brinkmann's isometric embedding \cite{liu2018rigidity}, because the mean curvatures are not invariant under orthochronous Lorentz transformation.

With the reference settled, we can now define the BY energy:
\begin{defn}
Given a closed spacelike conformally flat codimension-two surface $S$ embedded on $\Sigma$ in an $n$-dimensional spacetime, and $\tilde{\Sigma}$, which intersects with $\tilde{N}_p$ at $\tilde{S}$, satisfies the condition $\star$, then the Brown-York type energy is defined as
\beq
M_{BY}(S) := \frac{1}{\Omega_{n-2}(n-2)} \int_{S} \tilde{H} - H\; \dd \sigma
\eeq
where $ \sigma$ is the induced volume form on $S$, $H$ is the mean curvature of $S$ in $\Sigma$ and $\tilde{H}$ is the mean curvature of $\tilde{S}$ as embedded in $\tilde{\Sigma}$.
\end{defn}
\begin{rem}
The sign convention follows Shi and Tam \cite{shi2002positive}, which differs from the original definition in \cite{brown1993quasilocal,brown1999canonical} by an overall sign. This is because our normal $v$ is pointing outwards rather than inwards.
\end{rem}

The definition of the KELY energy is similar to the BY energy, but the $H$ is replaced by the norm of the mean curvature vector $|K|$. Therefore, unlike the BY energy, KELY energy is a covariantly defined QLE, so we do not need to fix any hypersurface $\Sigma$ or $\tilde{\Sigma}$ a priori.  Such a surface is shear-free, so the Gauss equation (Lemma \ref{lem:gauss}) implies
\beq
|\tilde{K}|^2=\frac{n-2}{n-3}\R.
\eeq
and we see that the norm of the mean curvature vector $|\tilde{K}|$ is fixed for any such isometric embedding, consistent with the rigidity of the lightcone embedding \cite{liu2018rigidity}.
\begin{defn}
Given a closed spacelike conformally flat codimension-two surface $S$ in an $n$-dimensional spacetime with spacelike mean curvature vector, the Kijowski-Epp-Liu-Yau type energy is defined as
\beq
M_{KELY}(S) := \frac{1}{\Omega_{n-2}(n-2)} \int_{S} |\tilde{K}| - |K|\; \dd \sigma
\eeq
where $\sigma$ is the induced volume form on $S$, $\tilde{K}$ is the mean curvature vector of S as embedded in $\M^n$.
\end{defn}
\begin{rem}
The original definitions of Epp and Kijowski-Liu-Yau differ in their references chosen. Strictly speaking, our definition here is actually closer to Epp's proposal.
\end{rem}
Note that Lemma \ref{lem:invariants} implies the mean curvature vector along lightcone cuts $S_l$ is spacelike for sufficiently small $l$ so $M_{KELY}(S_l)$ is defined.

\section{The small sphere limit in higher dimensions}

\subsection{The generalised Hawking-Hayward energy}
Here we prove a theorem that is slightly more general than Theorem \ref{thm:hawking},
\begin{thm}
Let $S_l$ be the family of surfaces shrinking towards $p$ along lightcone cuts defined with respect to $(p,e_0)$ in an $n$-dimensional spacetime, the limits of the Hawking-Hayward energy as $l$ goes to $0$ are
\begin{align}
\lim_{l\rightarrow 0} l^{-(n-1)} M_\alpha(S_l) =& \frac{T(e_0,e_0)}{n-1}, \label{hawkingnonvac}\\
\lim_{l\rightarrow 0} l^{-(n+1)} M_\alpha(S_l)|_{Ric=0}=&\frac{(6n^2-20n+8)E^2+6(n-3)H^2-3D^2}{36(n-2)(n-3)(n^2-1)}- \alpha\frac{(6n^2-8n-4)E^2+6nH^2-3D^2}{36(n-2)(n-3)(n^2-1)},\label{hawkingvac}
\end{align}
where the tensors $T,E,H,D$ are evaluated at $p$.
\end{thm}

In four dimensions, $D^2=4E^2, H^2=2B^2,  Q_0 = E^2+B^2$. We therefore recover from (\ref{hawkingnonvac},\ref{hawkingvac}) the results by Horowitz and Schmidt \cite{horowitz1982note} on the Hawking energy and also for the Hayward energy \cite{bergqvist1994energy}:
\begin{cor}
When $n=4$, the small sphere limits of the Hawking energy and Hayward energy are
\begin{align}
\lim_{l\rightarrow 0} \frac{M_0(S_l)}{l^3} = \frac{4\pi}{3}T_{00},\;\;\;\;\lim_{l\rightarrow 0} \frac{M_0(S_l)|_{Ric=0}}{l^5}  = \frac{1}{90}W,\\
\lim_{l\rightarrow 0} \frac{M_1(S_l)}{l^3} = \frac{4\pi}{3}T_{00},\;\;\;\;\lim_{l\rightarrow 0} \frac{M_1(S_l)|_{Ric=0}}{l^5}  =\frac{-1}{30}W.
\end{align}
\end{cor}
\begin{proof}

Using Proposition \ref{cutarea},  the area factor in front of the energy integral (\ref{Hawking}) will only contribute as 
\beq
\left(\frac{\text{Vol}(S_l)}{\Omega_{n-2}}\right)^{\frac{1}{n-2}} = l+O(l^2),
\eeq
 whereas the higher order contributions can be ignored as we are only interested in the leading order behaviour of the QLE. The volume form on $S_l$ is given as 
\beq 
 \dd\sigma=l^{n-2}\dd \Omega_{n-2}+O(l^n)\dd \theta_1\wedge\cdots\wedge\dd\theta_{n-2},
\eeq
 and we only ever need the flat part for our integrations later.

In non-vacuum, according to Lemma \ref{lem:invariants} the shear squared term has the leading order of the Weyl tensor squared which is of higher order than other terms.  Since it is beyond the order of consideration in non-vacuum, we can ignore its contribution here. Using Lemma \ref{lem:gauss}, we can simplify the HH energy integral
\beq
\begin{aligned}
M_\alpha(S_l)=&\frac{l^{n-1}}{\Omega_{n-2}(n-2)(n-3)}\int_{S_l}\dd\Omega_{n-2}\left( \frac{R}{2}+2Ric(\ell^+,\ell^-)-R(\ell^+,\ell^-,\ell^+,\ell^-)\right)+O(l^{n+1}),\\
=&\frac{l^{n-1}}{\Omega_{n-2}(n-2)(n-3)}\int_{S_l}\dd\Omega_{n-2}\left( \frac{R}{2}+2R_{+-}-R_{+-+-}\right)+O(l^{n+1}),\\
=&\frac{l^{n-1} G(e_0,e_0)}{(n-2)(n-1)}+O(l^{n+1}),\\
=&\frac{l^{n-1} \Omega_{n-2}}{(n-1)}T(e_0,e_0)+O(l^{n+1})
\end{aligned}
\eeq
where we have used Lemma \ref{lem:invariants} and the Einstein field equation (\ref{EFE}). Note that when applying the Gauss equation, all the curvature data above are evaluated on $S_l$. However, since only the leading order contribution here is relevant, we only keep the contribution due to the curvature data evaluated at the lightcone vertex $p$. This is not the case for the vacuum case as wel shall see.


For the vacuum case, we cannot omit the shear term,
\begin{align}
M_\alpha(S_l) =& \frac{l^{n-1}}{\Omega_{n-2}(n-2)(n-3)}\int_{S_l}(1-\alpha)\;\sigma^-_{ab}\sigma^{+ab} - C(\ell^+,\ell^-,\ell^+,\ell^-)(l)\;\dd\Omega_{n-2} + O(l^{n+3}).
\end{align}
Note that at this order we need to expand the Weyl tensor at $S_l$ using the data at $p$. We expand $C_{+-+-}(l)$ in a Taylor series around $p$. We then replace the partial derivatives with the covariant derivatives using (\ref{connection}) and obtain
\beq
C_{+-+-}(l) = C_{+-+-} + \nabla_+ C_{+-+-} l +\frac{l^2}{2} \nabla_+\nabla_+ C_{+-+-} - \frac{l^2}{3} C\indices{^\mu_{+-+}}C_{\mu +-+}.
\eeq
The covariant derivative terms will vanish after the integration over $S_l$ so we omit them as before. Together with the expansions of the null generators using Proposition \ref{lem:gauge}, we have
\beq
C(\ell^+,\ell^-,\ell^+,\ell^-)(l) = C_{+-+-} + \frac13 C\indices{^\mu_{+-+}}C_{\mu +-+}l^2.
\eeq
Using Lemma \ref{lem:invariants},\ref{lem:identities} and $P(Ric=0)=0$,
\beq
\begin{aligned}
M_\alpha(S_l) =& \frac{l^{n-1}}{\Omega_{n-2}(n-2)(n-3)}\int_{S_l}\left((1-\alpha)\;\sigma^-_{ab}\sigma^{+ab} - \frac13l^2 C\indices{^\mu_{+-+}}C_{\mu +-+}\right)\dd\Omega_{n-2} + O(l^{n+3}),\\
=& \frac{l^{n+1}}{18\Omega_{n-2}(n-2)(n-3)}\int_{S_l} \left(C_{\mu+\nu+}C\indices{^\mu_+^\nu_+}+4C_{\mu+\nu-}C\indices{^\mu_+^\nu_+}-10C_{\mu+-+}C\indices{^\mu_+_-_+}\right)\;\dd\Omega_{n-2}\\
 &- \frac{\al l^{n+1}}{18\Omega_{n-2}(n-2)(n-3)}\int_{S_l} \left(C_{\mu+\nu+}C\indices{^\mu_+^\nu_+}+4C_{\mu+\nu-}C\indices{^\mu_+^\nu_+}-4C_{\mu+-+}C\indices{^\mu_+_-_+}\right)\;\dd\Omega_{n-2} + O(l^{n+3}),\\
=&\frac{(6n^2-20n+8)E^2+6(n-3)H^2-3D^2}{36(n-2)(n-3)(n^2-1)}l^{n+1}- \alpha\frac{(6n^2-8n-4)E^2+6nH^2-3D^2}{36(n-2)(n-3)(n^2-1)}l^{n+1}+ O(l^{n+3}).
\end{aligned}
\eeq
Set $\al=0$  gives $\Q$ as stated in Theorem \ref{thm:hawking}.
\end{proof}

\subsection{The generalised Brown-York energy}

\begin{proof}{\it of Theorem \ref{thm:BY}}

$\tilde{H}$ as embedded in $\tilde{\Sigma}$ is given by (\ref{vacuumgauss})
\beq
\tilde{H}=\frac{n-2}{n-3}\frac{\R}{2\theta^+}+\frac{\theta^+}{2}.
\eeq
Using Lemma \ref{lem:outerexp},\ref{lem:innerexp},\ref{lem:ricci} we have
\beq
\begin{aligned}
\tilde{H}-H =& \frac{n-2}{n-3}\frac{\R}{2\theta^+}+\theta^-,\\
=&\frac{l}{n-3}\left(\frac{R}{2}+2R_{+-}-R_{+-+-}\right) +P(Ric)l^3 \\
&-\frac{l^3}{9(n-3)} \left(5C\indices{^\mu_{+-+}}C_{\mu +-+}-2C\indices{^\mu_+^\nu_+}C_{\mu +\nu -}-\frac12C\indices{^\mu_+^\nu_+}C_{\mu +\nu +}\right)l^3+ O(l^5)
\end{aligned}
\eeq
where we have separated the relevant terms for non-vacuum and vacuum cases into two lines.

In non-vacuum, using Lemma \ref{lem:identities}, we obtain
\beq
\begin{aligned}
M_{BY}(S_l)=&\frac{1}{\Omega_{n-2}(n-2)}\int_{S_l} \frac{l}{n-3}\left(\frac{R}{2}+2R_{+-}-R_{+-+-}\right)+O(l^3)\;\dd\sigma, \\
=&\frac{1}{\Omega_{n-2}(n-2)}\int_{S_l} \frac{l^{n-1}}{n-3}\left(\frac{R}{2}+2R_{+-}-R_{+-+-}\right)\dd\Omega_{n-2}+O(l^{n+1}), \\
=&\frac{l^{n-1}}{(n-2)(n-3)}\left(\frac{R}{2}+\frac{(n-2)Ric(e_0,e_0)-R}{n-1}-\frac{Ric(e_0,e_0)}{n-1}\right)+O(l^{n+1}),\\
=&\frac{l^{n-1}}{(n-2)(n-1)}G(e_0,e_0)+O(l^{n+1}),\\
=&\frac{l^{n-1}\Omega_{n-2}}{(n-1)}T(e_0,e_0)+O(l^{n+1}).
\end{aligned}
\eeq
where the Einstein field equation (\ref{EFE}) is used in the last line.

In vacuum, $P(Ric=0)=0$, again using Lemma \ref{lem:identities},
\beq
\begin{aligned}
M_{BY}(S_l)=&\frac{1}{\Omega_{n-2}(n-2)}\int_{S_l}\frac{-l^3}{9(n-3)} \left(5C\indices{^\mu_{+-+}}C_{\mu +-+}-2C\indices{^\mu_+^\nu_+}C_{\mu +\nu -}-\frac12C\indices{^\mu_+^\nu_+}C_{\mu +\nu +}\right)+ O(l^5)\;\dd\sigma, \\
=&\frac{1}{\Omega_{n-2}(n-2)}\int_{S_l}\frac{-l^{n+1}}{9(n-3)} \left(5C\indices{^\mu_{+-+}}C_{\mu +-+}-2C\indices{^\mu_+^\nu_+}C_{\mu +\nu -}-\frac12C\indices{^\mu_+^\nu_+}C_{\mu +\nu +}\right)\dd\Omega_{n-2}+ O(l^{n+3}), \\
=& \frac{(8 - 20 n + 6 n^2)E^2 + 6 (n-3)H^2 -3 D^2}{36(n-3)(n-2)(n^2-1)}l^{n+1}+ O(l^{n+3}),\\
=& \Q l^{n+1}+ O(l^{n+3}).
\end{aligned}
\eeq
\end{proof}

\subsection{The generalised Kijowski-Epp-Liu-Yau energy}
\begin{proof}{\it of Theorem \ref{thm:LY}}

The Gauss equation applied on a Minkowski lightcone $\tilde{N}_p$ gives
\beq
|\tilde{K}|^2=\frac{n-2}{n-3}\R.
\eeq
Using Lemma  \ref{lem:ricci},\ref{lem:invariants}, we have
\beq
\begin{aligned}
|\tilde{K}|-|K|=&\left(\frac{n-2}{n-3}\R\right)^\frac12 - \left(-2\theta^+\theta^-\right)^\frac12,\\
=&\frac{(n-2)}{l}+\frac{l}{2(n-3)}\left(R+\frac{4n}{3}R_{+-}-\frac{n(n-1)}{3}R_{+-+-}\right)-\frac{l^3}{288(n-2)}\left(\frac{n(n-1)C_{+-+-}}{n-3}\right)^2\\
&+\frac{l^3}{90(n-3)} \left[(4 C_{\mu+\nu-}C\indices{^\mu_+^\nu_+} + C_{\mu+\nu+}C\indices{^\mu_+^\nu_+}) (2 n-1) + (4 - 9 n - 3 n^2)C\indices{^\mu_{+-+}}C_{\mu +-+}\right]\\
&-\frac{(n-2)}{l}-\frac{l}{6}\left[4R_{+-} - (n+2)R_{+-+-}\right] +\frac{l^3(n+2)^2}{288(n-2)}C_{+-+-}^2\\
&+\frac{l^3}{15}\left(\frac{n+6}{2}C\indices{^\mu_{+-+}}C_{\mu +-+}-\frac43C\indices{^\mu_+^\nu_+}C_{\mu +\nu -}-\frac13C\indices{^\mu_+^\nu_+}C_{\mu +\nu +}\right)+P(Ric) l^3+O(l^3),\\
=&\frac{l}{(n-3)}\left(\frac12R+2R_{+-}-R_{+-+-}\right)-\frac{l^3(n^2-n-3)}{24(n-2)(n-3)^2}C_{+-+-}^2\\
&-\frac{l^3}{9(n-3)} \left(5C\indices{^\mu_{+-+}}C_{\mu +-+}-2C\indices{^\mu_+^\nu_+}C_{\mu +\nu -}-\frac12C\indices{^\mu_+^\nu_+}C_{\mu +\nu +}\right)l^3+P(Ric) l^3+O(l^3).
\end{aligned}
\eeq
Using Lemma \ref{lem:identities}, we have in non-vacuum,
\beq
\begin{aligned}
M_{KELY}(S_l)=&\frac{1}{\Omega_{n-2}(n-2)}\int_{S_l} \frac{l^{n-1}}{n-3}\left(\frac{R}{2}+2R_{+-}-R_{+-+-}\right)\dd\Omega_{n-2}+O(l^{n+1}), \\
 =&\frac{l^{n-1}}{(n-2)(n-1)}G(e_0,e_0)+O(l^{n+1}),\\
=&\frac{l^{n-1}\Omega_{n-2}}{(n-1)}T(e_0,e_0)+O(l^{n+1}),
\end{aligned}
\eeq
This is again the same small sphere behaviour as the HH energy and BY energy in non-vacuum. However, in vacuum we will see an extra term that is proportional to $E^2$:
\beq
\begin{aligned}
M_{KELY}(S_l)=&\Q\, l^{n+1}-\frac{l^{n+1}(n^2-n-3)}{24\Omega_{n-2}(n-2)^2(n-3)^2}\int_{S_l}C_{+-+-}^2\dd\Omega_{n-2}+ O(l^{n+3}), \\
=&\Q\, l^{n+1} -\frac{l^{n+1}(n^2-n-3)E^2}{12(n^2-1)(n-2)^2(n-3)^2}+ O(l^{n+3}).
\end{aligned}
\eeq
\end{proof}

\section{Discussion}
We discover a new quantity $\Q$ represented in terms of the electromagnetic decompositions of the Weyl tensor. It replaces the BR superenergy $Q$ in the context of QLE. Albeit that $\Q$ is not positive-definite, we believe that $\Q$ nevertheless characterises the local gravitational energy content, perhaps in an elusive way. Further investigation is certainly needed to clarify its physical meaning. One can try to study other QLE proposals. The Wang-Yau energy \cite{wang2009quasilocal} is the most recent QLE proposal using the Hamilton-Jacobi approach. It overcomes all the shortcomings of the BY and the LY proposals \cite{miao2015quasi}. However, its potential generalisation to higher dimensions is obscure as the definition relies on the data of isometric embedding, so the existence alone is not enough~\cite{wang2009quasilocal,chen2018evaluating} even for the small sphere calculations. It would be insightful to see how the Wang-Yau proposal can be extended to higher dimensions and whether $\Q$ appears in the limit. There is another QLE that can be naturally generalised to higher dimensions, the Bartnik energy \cite{bartnik1989new}. Roughly speaking, it is defined as the infimum among the ADM masses evaluated on all horizon-free asymptotically flat extensions of a compact spacetime domain. It has many desirable features such as positivity and monotonicity, so it has always been a promising candidate of QLE. However, the Bartnik mass is difficult to evaluate generally, especially in the spacetime case. The small sphere limit of the Bray's modification, the Bartnik-Bray mass, has been evaluated in a time-symmetric (Riemannian) setting \cite{wiygul2018bartnik}, but they are not comparable with results in the spacetime setting. It would be very interesting to know how to compute its small sphere limit along lightcone cuts and compare to our results here. 

\ack 
 We thank Jos\'e M.M.Senovilla for clarifying the Bel-Robinson superenergy and Renato Renner for useful comments. This work is supported by the Swiss National Science Foundation via
the National Center for Competence in Research “QSIT”, and by the Air Force Office of Scientific Research (AFOSR) via grant FA9550-16-1-0245. 
\\
\appendix

\section{Proof of Lemma 3.14}\label{app:proof1}
We restate Lemma \ref{lem:identities} here
\begin{align}
\int_{S^{n-2}}  \dd \Omega_{n-2} n^in^j =& \frac{\Omega_{n-2}}{n-1}\delta_{ij},\\
\int_{S^{n-2}}  \dd \Omega_{n-2} n^in^jn^kn^l =& \frac{\Omega_{n-2}}{n^2-1}(\delta_{ij}\delta_{kl}+\delta_{ik}\delta_{jl}+\delta_{il}\delta_{jk}),\\
\int_{S^{n-2}}  \dd \Omega_{n-2} R_{+-+-} =& \frac{\Omega_{n-2}}{n-1}Ric(e_0,e_0),\\
\int_{S^{n-2}}  \dd \Omega_{n-2} C_{+-+-} =& 0,\\
\int_{S^{n-2}}  \dd \Omega_{n-2} R_{++} =& \Omega_{n-2}\frac{nRic(e_0,e_0)+R}{n-1},\\
\int_{S^{n-2}}  \dd \Omega_{n-2} R_{+-} =& \Omega_{n-2}\frac{(n-2)Ric(e_0,e_0)-R}{2(n-1)},\\
\int_{S^{n-2}}  \dd \Omega_{n-2} C_{\mu+\nu+}C\indices{^\mu_+^\nu_+} =& \Omega_{n-2}\left(\frac{2(n^2+2)E^2+6nH^2+3D^2}{2(n^2-1)}\right),\\
\int_{S^{n-2}}  \dd \Omega_{n-2} C_{\mu+\nu-}C\indices{^\mu_+^\nu_+} =& \Omega_{n-2}\left(\frac{2(n^2-4)E^2+6H^2-3D^2}{4(n^2-1)}\right),\\
\int_{S^{n-2}}  \dd \Omega_{n-2} C_{\mu+-+}C\indices{^\mu_+_-_+} =& \Omega_{n-2}\left(\frac{2(n-1)E^2+3H^2}{2(n^2-1)}\right),\\
\int_{S^{n-2}}  \dd \Omega_{n-2} C_{+-+-}^2 =&  \Omega_{n-2}\frac{2E^2}{n^2-1}
\end{align}
where $\dd\Omega_{n-2}$ is the natural volume form on a unit $S^{n-2}$ and $\Omega_{n-2}:=\int_{S^{n-2}}\dd\Omega_{n-2}$ and $n^i$ are normalised spacelike vectors. 

\begin{proof}
The first two identities are two specific cases of a more general result~\cite{othmani2011polynomial}
\beq
\begin{aligned}\label{eqn:evenn}
\int_{S^n} \dd \Omega_{n} n^{i_1}\dots n^{i_k}&= \frac{\Omega_n(n-1)!!}{(n+k-1)!!}\delta^{(k)}_{i_1i_2...i_k},  \;\;\,\mbox{($k$ even)}\\
\int_{S^n} \dd \Omega_{n} n^{i_1}\dots n^{i_k}&= 0, \hspace{3cm} \mbox{($k$ odd)}
\end{aligned}
\eeq
where $\delta^{(k)}_{i_1i_2...i_k}$ is defined recursively for even $k$: 
\begin{align}
 \delta^{(k+2)}_{i_1i_2...i_{k+2}}&=(k+1)!\delta_{i(j}^{(2)}\delta^{(k)}_{i_1i_2...i_k)}=\delta_{ij}\delta^{(k)}_{i_1i_2...i_k}+\delta_{ii_1}\delta^{(k)}_{ji_2...i_k}+\dots+\delta_{ii_k}\delta^{(k)}_{i_1i_2...j}.
\end{align}

 $\delta^{(2)}_{ij}$ is the usual Kronecker delta $\delta_{ij}$ and the second equality above is due to the fact that $\delta^{(k)}_{i_1i_2...i_k}$ so defined is totally symmetric. Then (A.1, A.2) follows from
\beq
\delta^{(2)}_{ij}=\delta_{ij},\;\;\delta^{(4)}_{klmn}=\delta_{kl}\delta_{mn}+\delta_{km}\delta_{ln}+\delta_{kn}\delta_{ml}.
\eeq
 (A.1, A.2) will be used in all the rest integrals.
 
Considering (A.3), using Proposition \ref{lem:gauge},
\beq
\int_{S^{n-2}}  \dd \Omega_{n-2} R_{+-+-} = \int \dd \Omega_{n-2} R_{oioj}n^in^j = \frac{\Omega_{n-2}}{n-1}R\indices{_0_i_0^i}=\frac{\Omega_{n-2}}{n-1}Ric(e_0,e_0).
\eeq
(A.4) follows from (A.3) by setting $Ric=0$. For (A.5,A.6),
\beq
\begin{aligned}
\int_{S^{n-2}}  \dd \Omega_{n-2} R_{++}=&\int_{S^{n-2}}  \dd \Omega_{n-2} R_{00}+R_{ij}n^in^j=\Omega_{n-2}R_{00}+\frac{\Omega_{n-2}}{n-1}(R+R_{00})=\Omega_{n-2}\frac{nRic(e_0,e_0)+R}{n-1},\\
\int_{S^{n-2}}  \dd \Omega_{n-2} R_{+-}=&\frac12\int \dd \Omega_{n-2} R_{00}-R_{ij}n^in^j = \frac12\Omega_{n-2}R_{00}-\frac12\frac{\Omega_{n-2}}{n-1}(R+R_{00})=\Omega_{n-2}\frac{(n-2)Ric(e_0,e_0)-R}{2(n-1)}
\end{aligned}
\eeq
For (A.7-A.9),  we first expand the integrands
\beq
\begin{aligned}
C_{\mu+\nu+}C\indices{^\mu_+^\nu_+}=&E^2-2E\indices{_i^k}E\indices{^k_j}n^in^j+2E_{ij}D\indices{^i_l^j_k}n^ln^k+E_{ij}E_{lk}n^in^jn^ln^k\\
&+4H\indices{^k^l_i}H_{(kl)j}n^in^j-2H\indices{_i^m_j}H_{lmk}n^in^jn^ln^k+D\indices{^m_i^n_j}D_{mlnk}n^in^jn^ln^k,\\
C_{\mu+\nu+}C\indices{^\mu_+^\nu_-}=&\frac12E^2+\frac12(2H\indices{_i^m_j}H_{kml}-E_{ij}E_{kl}-D\indices{^m_i^n_j}D_{mknl})n^in^jn^kn^l,\\
C_{\mu+-+}C\indices{^\mu_+_-_+}=&E_{ik}E\indices{^k_j}n^in^j-E_{ij}E\indices{_k_l}n^in^jn^kn^l+H_{iqk}H\indices{_j^q_l}n^in^jn^kn^l.
\end{aligned}
\eeq
Then we need to apply Lemma \ref{lem:emidentity} to compute the integrals,
\beq
\begin{aligned}
\int_{S^{n-2}}  \dd \Omega_{n-2} C_{\mu+\nu+}C\indices{^\mu_+^\nu_+} =&\Omega_{n-2}\left(E^2-\frac{2E^2}{n-1}+\frac{2E^2}{n-1}+\frac{2E^2}{n^2-1}+\frac{3H^2}{n-1}-\frac{3H^2}{n^2-1}+\frac{3D^2+2E^2}{2(n^2-1)}\right)\\
=&\Omega_{n-2}\left(\frac{2(n^2+2)E^2+6nH^2+3D^2}{2(n^2-1)}\right),\\
\int_{S^{n-2}}  \dd \Omega_{n-2} C_{\mu+\nu-}C\indices{^\mu_+^\nu_+} =&\frac{\Omega_{n-2}}{2(n^2-1)}\left(E^2(n^2-1)-2E^2+3H^2-\frac32D^2-E^2\right) =\Omega_{n-2}\left(\frac{2(n^2-4)E^2+6H^2-3D^2}{4(n^2-1)}\right),\\
\int_{S^{n-2}}  \dd \Omega_{n-2} C_{\mu+-+}C\indices{^\mu_+_-_+} =&\Omega_{n-2}\left(\frac{E^2}{n-1}-\frac{2E^2}{n^2-1}+\frac{3H^2}{2(n^2-1)}\right)=\Omega_{n-2}\left(\frac{2(n-1)E^2+3H^2}{2(n^2-1)}\right).
\end{aligned}
\eeq
The last identity (A.10) follows from
\begin{align}
 \int_{S^{n-2}}  \dd \Omega_{n-2} C^2_{+-+-}= \int \dd \Omega_{n-2} \;E_{ij}E_{kl}n^in^jn^kn^l=\Omega_{n-2}\frac{2E^2}{n^2-1}.
\end{align}
\end{proof}

\section{The Misner-Sharp energy}\label{app:MS}
 In a spherically symmetric spacetime, the Misner-Sharp energy is defined for codimension-2 sphere of symmetry with area radius $r$. Generally, given the metric of an $n$-dimensional spherically symmetric spacetime
\beq
g_{\mu\nu} = h_{\mu\nu}\dd z^\mu \dd z^\nu + r(z)q_{ij}\dd x^i \dd x^j.
\eeq
One can define the Misner-Sharp energy as 
\beq
\begin{aligned}
M_{MS}(r)&=\frac{r^{n-3}}{2}(1-\nabla_a r \nabla^a r),\\
&=\frac{r^{n-1}}{2(n-2)}\left(\frac{n-2}{r^2}-\frac{K^2}{n-2}\right),\\
&=\frac{r^{n-1}}{2(n-2)(n-3)}\left(\R-\frac{n-3}{n-2}K^2\right)
\end{aligned}
\eeq
where $\R=\frac{(n-2)(n-3)}{r^2}$ because of spherical symmetry and $K_a=-\frac{(n-2)}{r}\nabla_a r$.\\

\section*{References}
\bibliographystyle{unsrt}

\bibliography{HM}

\end{document}